\documentclass[prl,twocolumn,superscriptaddress]{revtex4}

\usepackage[nice]{nicefrac}
\usepackage{array}
\usepackage{amsmath}
\usepackage{graphicx}
\usepackage{enumerate,amsthm,amssymb,color}
\usepackage{bbm}
\usepackage{bm}
\usepackage{synttree}
\usepackage{caption}
\usepackage{subfig}
\usepackage{multirow}

\newcommand{\Mod}[1]{\ (\mathrm{mod}\ #1)}

\newcommand{\ket}[1]{|#1\rangle}
\newcommand{\bra}[1]{\langle#1|}

\newcommand{\COMMENT}[1]{}

\newtheorem*{conjecture*}{Conjecture}

\graphicspath{ {images/} }

\begin{document}

\title{Tsirelson's bound and Landauer's principle in a single-system game}

\author{Luciana Henaut}\email{luciana.henaut.14@ucl.ac.uk}\affiliation{Department of Physics and Astronomy, University College London, London, WC1E 6BT, UK}
\author{Lorenzo Catani}\email{lorenzo.catani.14@ucl.ac.uk}\affiliation{Department of Physics and Astronomy, University College London, London, WC1E 6BT, UK}\author{Dan E. Browne}\affiliation{Department of Physics and Astronomy, University College London, London, WC1E 6BT, UK}\author{Shane Mansfield}\affiliation{Sorbonne Universit\'e, CNRS, Laboratoire d'Informatique de Paris 6, F-75005 Paris, France}
\author{Anna Pappa}\affiliation{Department of Physics and Astronomy, University College London, London, WC1E 6BT, UK}

\date{\today}

\begin{abstract}
\noindent We introduce a simple single-system game inspired by the Clauser-Horne-Shimony-Holt (CHSH) game. For qubit systems subjected to unitary gates and projective measurements, we prove that any strategy in our game can be mapped to a strategy in the CHSH game, which implies that  Tsirelson's bound also holds in our setting.
More generally, we show that the optimal success probability depends on the reversible or irreversible character of the gates, the quantum or classical nature of the system and the system dimension.  We analyse the bounds obtained in light of Landauer's principle, showing the entropic costs of the erasure associated with the game. This demonstrates a connection between the reversibility in fundamental operations embodied by Landauer's principle and Tsirelson's bound, that arises from the restricted physics of a unitarily-evolving single-qubit system. 
\end{abstract}

\maketitle

Computational protocols in which quantum mechanical strategies provide an advantage over classical ones have long been an important focus of study. A way of recasting the Clauser-Horne-Shimony-Holt formulation  \cite{CHSH69} of Bell's celebrated theorem \cite{Bell1966} into a game for which quantum strategies can provide an advantage has been proposed in \cite{vanDam} and has since been referred to as the CHSH game. The players of the game, Alice and Bob, are separated and unable to communicate with each other; each is given one randomly uniform bit, labeled $a$ and $b$ respectively, and they win the game if they return single bits, $x$ and $y$ respectively, such that $x\oplus y =a\cdot b \pmod 2$. 

In game theory, the optimal success probability for a game is called its \textit{value}, which we denote by $\omega$. The value of the CHSH game, $\omega(\textrm{CHSH})$, depends upon the physics of the systems exploited by Alice and Bob. Famously, if Alice and Bob employ only classical strategies,  the value of the CHSH game is $\omega(\textrm{CHSH})=0.75$. On the other hand, if they have access to quantum resources, $\omega(\textrm{CHSH})=\cos^2(\frac{\pi}{8})\approx0.85$. The limitation on the value of the game for classical systems is called a Bell inequality, and the value $0.75$ is often called the Bell bound. The fact that the value of the game when using quantum resources violates the Bell bound, but is nevertheless limited substantially below 1, was first noted by Tsirelson \cite{Cirelson1980}, and the value $\cos^2(\frac{\pi}{8})$ is known as Tsirelson's bound. Popescu and Rohrlich \cite{Popescu1994} noted that in more general theories than quantum mechanics, perfect strategies for the CHSH game that achieve a value of 1 could exist via a correlation now known as a Popescu-Rohrlich (PR) box, without violating the no-signaling assumption between Alice and Bob during the game.

The CHSH game is of great importance because the dependence of its value from the underlying physical model gives us a tool to distinguish different types of theories experimentally, and allows us to test nature. It also reveals insights into a non-classical feature of quantum mechanics (known colloquially as ``non-locality''), which has proven to be a resource for quantum technologies, such as device independent cryptography \cite{Pironio10}. Generalisations to $\bmod{q}$ arithmetics have also been proposed  \cite{Buhrman2005,Ji2008,Liang2009,Bavarian2015}. Naturally, a key focus of these studies has been to find the classical (Bell bound) and quantum value (Tsirelson bound) for these \textit{$CHSH_q$ games}. However, success has been limited.
Upper bounds for the quantum value given by a precise mathematical expression have been provided in \cite{Bavarian2015} when $q$ is a prime or prime power, but these are not known to be tight. Moreover, numerical analysis on lower and upper bounds suggest different values \cite{Liang2009}.

Following \cite{Dunjko2016, Barz2016,Clementi2017}, we propose and investigate  a simple one-player variant of the CHSH game that uses a single system as resource. Due to its similarity with the CHSH game we call it the \textit{CHSH* game}. However, unlike the CHSH game that involves two space-like separated parties, the CHSH* game cannot involve any non-locality argument to explain the computational advantages. Similarly, it does not show any contextuality (at least in its usual formulations \cite{KochenSpecker67, Spekkens2005}), which in other computational settings is known to be necessary for non-linear computations \cite{Raussendorf13}.

We first show that, when the player applies unitary dynamics and projective measurements on a qubit system, the value of the CHSH* game is equal to Tsirelson's bound; this is proven via an explicit mapping between the strategies in the CHSH* and CHSH games. We then show that this setting is sensitive to a broad range of properties of the system used, specifically whether the system is quantum or classical, what is the set of operations allowed to the player (namely reversible versus irreversible and Clifford versus non-Clifford) and what is the dimension of the space in which the system resides. Following Landauer's assertion that only reversible operations are truly fundamental, we show that bit erasure is a powerful tool for increasing the winning probability, shedding light on the source of quantum advantage in this game. We finally conjecture that our results also apply to the CHSH$_q^*$ game for any dimension $q,$ by considering the case of $q=3.$

\begin{figure}[h!]
\includegraphics[width=.45\textwidth]{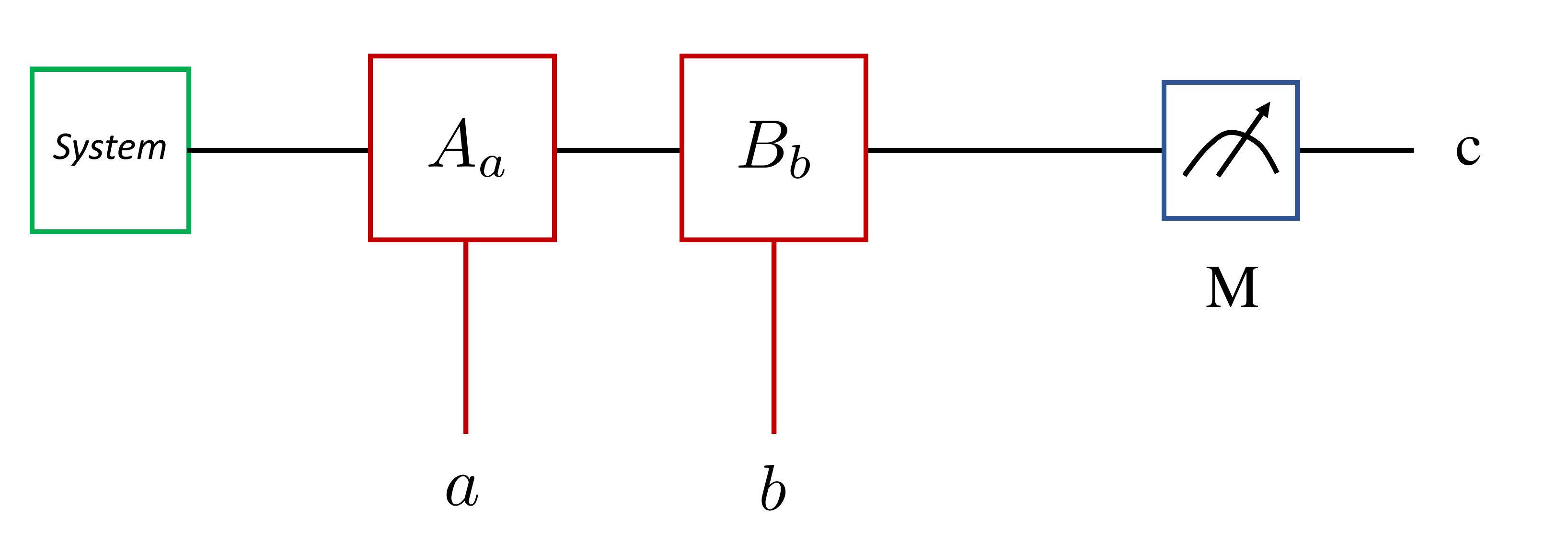}
\caption{\footnotesize{\textbf{Single-system protocol.} An initial system is subjected to controlled transformations,
with control bits $a$ and $b$, respectively, and then measured. The goal is to maximise the probability that the value of the output is the product of the values of the input bits. }}
\label{Protocol}
\end{figure}

\paragraph{The CHSH* game.} In this game (illustrated in Fig.~\ref{Protocol}), a single player has in her possession a single system of dimension $d$, that can be classical or quantum. She is given a specification of the state preparations, transformations and measurements that she is allowed to employ and in the course of the game, she is also provided with two uniformly random bits $a$ and $b$. Choosing from the allowed operations, the player must specify in advance an initial state, controlled operations $A_a$ and $B_b$  and a final two-outcome measurement $M$. Once the player receives $a$ and $b$, the corresponding operations are implemented in sequence and measurement $M$ is performed, returning outcome $c$. The player wins the game when $c=a\cdot b \pmod{2}$. We are interested in finding the value $\omega(\textrm{CHSH*})$ of this game, which corresponds to the average winning probability over all possible strategies: 
\begin{equation*}\label{success_pr}
\omega(\textrm{CHSH*}) = \max_{\textrm{all strategies}} \frac{1}{4}\sum_{a,b\in\mathbb{Z}_2} p{\text{($c=a\cdot b \mid a,b$)}}.
\end{equation*}


\begin{table*}[t!] 
\centering

\begin{tabular}{|l|l|l|l|l|c|}
\hline
Name of setting      & System Type & Initial states &  Transformations &  Measurements   & $\omega$(CHSH*)                \\ \hline\hline
Unitary              & Quantum     & Any                    & Any unitary gate        & Any two-outcome PVM & $\cos^2(\frac{\pi}{8})$ \\ \hline
Clifford             & Quantum     & Pauli eigenstates      & Clifford group gates    & Pauli measurements                     & 0.75 \\ \hline
Reversible Classical & Classical   & Any                    & Reversible gates        & n/a                                    & 0.75 \\ \hline
Irreversible         & Classical/Quantum     & Any                    & Any                     & Any                                    & 1\\ \hline
\end{tabular}
\caption{\footnotesize{The different settings of the CHSH* game for systems of dimension $d=2$.}}
\label{settings}
\end{table*}


\textit{Relationship with the CHSH game.} In this work, we will study the CHSH* game in a variety of \textit{settings} (see Table~\ref{settings}), where we make different assumptions about the physics of the system in which the game is cast. First, we consider the case where the player's system is a single qubit in the \textit{unitary setting}, meaning that all transformations applied during the game are unitary. We further assume that the final measurement is a projective two-outcome measurement.

\newtheorem{Prop}{Proposition}
\begin{Prop}\label{MainTheorem}
The value of the CHSH* game with a $d=2$ quantum system in the unitary setting is $\cos^2(\frac{\pi}{8})$.
\end{Prop}

\noindent{}This result follows directly from the following lemma.

\newtheorem{Lemma}{Lemma}
\begin{Lemma}\label{MainLemma}
For every strategy in the CHSH* game in the unitary setting with $d=2$, we can derive an equivalent strategy for the two-player CHSH game such that both strategies lead to the same average success probability.
\end{Lemma}

\begin{proof}

We prove this explicitly. We first consider the CHSH* game and assume without loss of generality that the initial state is $\ket{+}$ and the measurement is the Pauli $X$ observable. A strategy thus consists of optimally choosing the gates $A_0,A_1,B_0,B_1$.

In Fig.~\ref{CHSH}, we show how, given a strategy for the CHSH* game, we can construct a strategy for the CHSH game. The key ingredient is a teleportation protocol that uses entanglement shared via the CNOT gate to teleport the effect of gate $A_a$ from one site (Alice's) to another spatially separated site (Bob's). Since operations $A_a$ are unitary, it holds that
\[A_a^T\otimes \mathbb{I}\Big(\frac{\ket{00}+\ket{11}}{\sqrt{2}}\Big)= \mathbb{I}\otimes A_a\Big(\frac{\ket{00}+\ket{11}}{\sqrt{2}}\Big) \, .\]
The teleported state on Bob's side after Alice measures her qubit is $A_a Z^{x}\left|+\right\rangle,$ where $Z$ is the Pauli $Z.$ The bits $x$ and $y$ are Alice's and Bob's outputs respectively.
In order to prove the lemma, we will show that the success probabilities for obtaining $c=a\cdot b$ in the CHSH* game and $x \oplus y = a\cdot b$ in the CHSH game are equal, i.e.:
\begin{equation*}\label{SuccessProb}\sum_{a,b}\Pr{(c=a\cdot b |a,b)}=\sum_{a,b}\Pr{(x\oplus y = a\cdot b |a,b)}.\end{equation*}
We proceed by showing that the terms in the above sums are pairwise equal, i.e. for every $a,b\in\{0,1\}$,
\begin{equation*}\label{PairWise}\Pr{(c=a\cdot b \mid a,b)}=\Pr{(x\oplus y = a\cdot b \mid a,b)}.\end{equation*}
In the case that $x=0$ this holds trivially; and when $x=1$, this reduces to showing that
\begin{eqnarray}
|\bra{+}B_bA_a\ket{+}|^2 &= &|\bra{-}B_bA_a\ket{-}|^2 \nonumber\\
|\bra{-}B_bA_a\ket{+}|^2 &= &|\bra{+}B_bA_a\ket{-}|^2,\nonumber
\end{eqnarray} 
which is necessarily true for any $2\times 2$ unitary gates.
\end{proof}

To see that Lemma~\ref{MainLemma} implies Proposition \ref{MainTheorem} we recall that Tsirelson's bound upperbounds the CHSH game at probability $\cos^2(\frac{\pi}{8})\approx0.85.$ A strategy which achieves this success probability involves the following gates: $A_0=\openone,A_1=S, B_0=T^\dagger, B_1=T,$ where $S  = R_z (\frac{\pi}{2})$ and $T  = R_z (\frac{\pi}{4})$ correspond to rotations around the $z$-axis in the usual Bloch sphere representation of the qubit. These unitaries are the gates mapping between the observables typically used to attain the Tsirelson bound in the CHSH game when the parties share {a Bell pair}. This strategy is also strictly related to the optimal strategies used in other tasks involving one qubit, like quantum random access codes \cite{Galvao} and parity oblivious multiplexing \cite{Spekkens2009}.  Lemma~\ref{MainLemma} demonstrates a tight link between Tsirelson's bound for the CHSH game and the value of CHSH* game in the above setting.

\begin{figure*}[t!]
\centering
\subfloat[][\label{MappingSingle}]
{\includegraphics[width=.47\textwidth,height=.21\textheight]{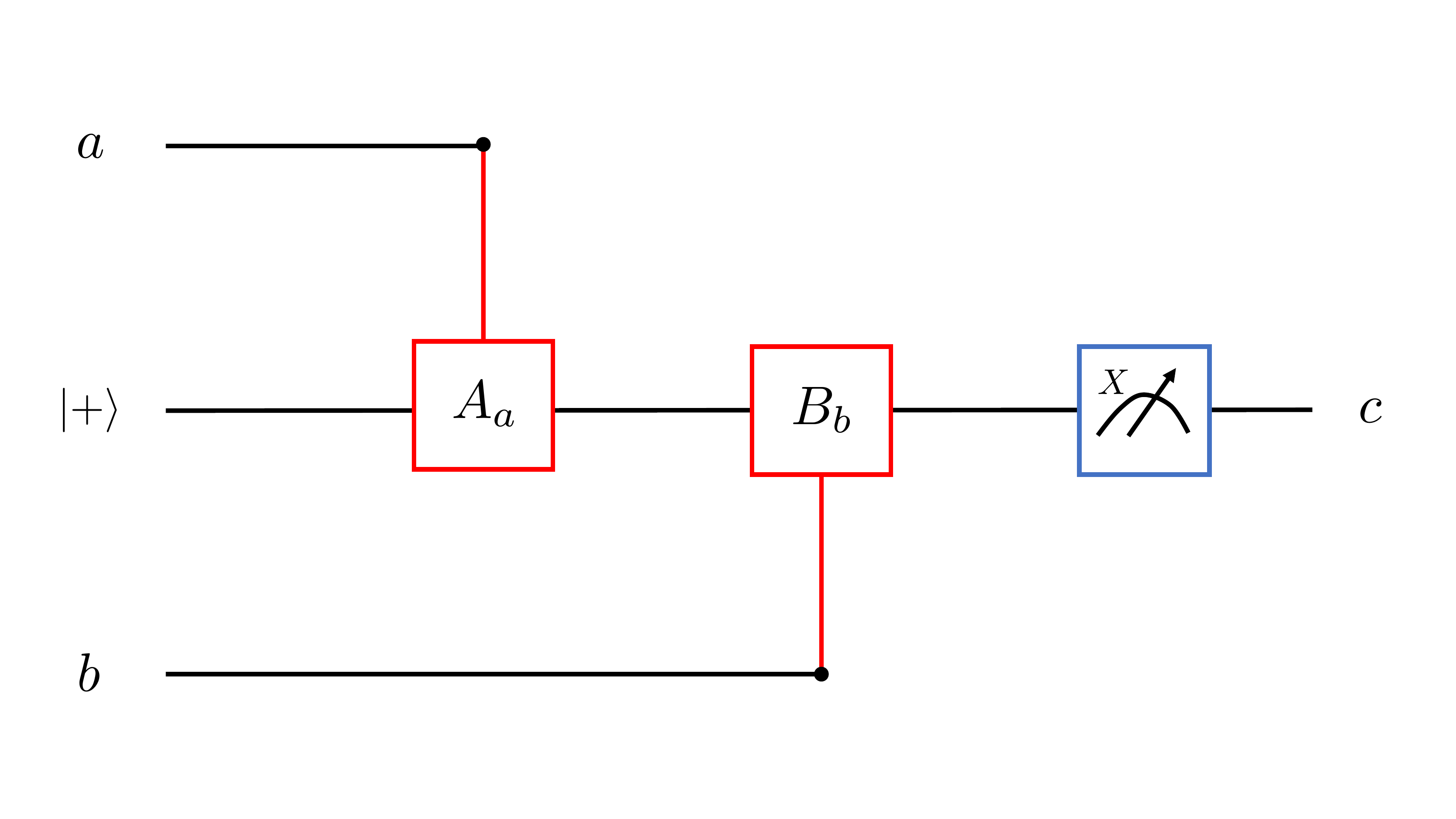}}\hfill
\subfloat[][\label{MappingCHSH}]
{\includegraphics[width=.49\textwidth,height=.21\textheight]{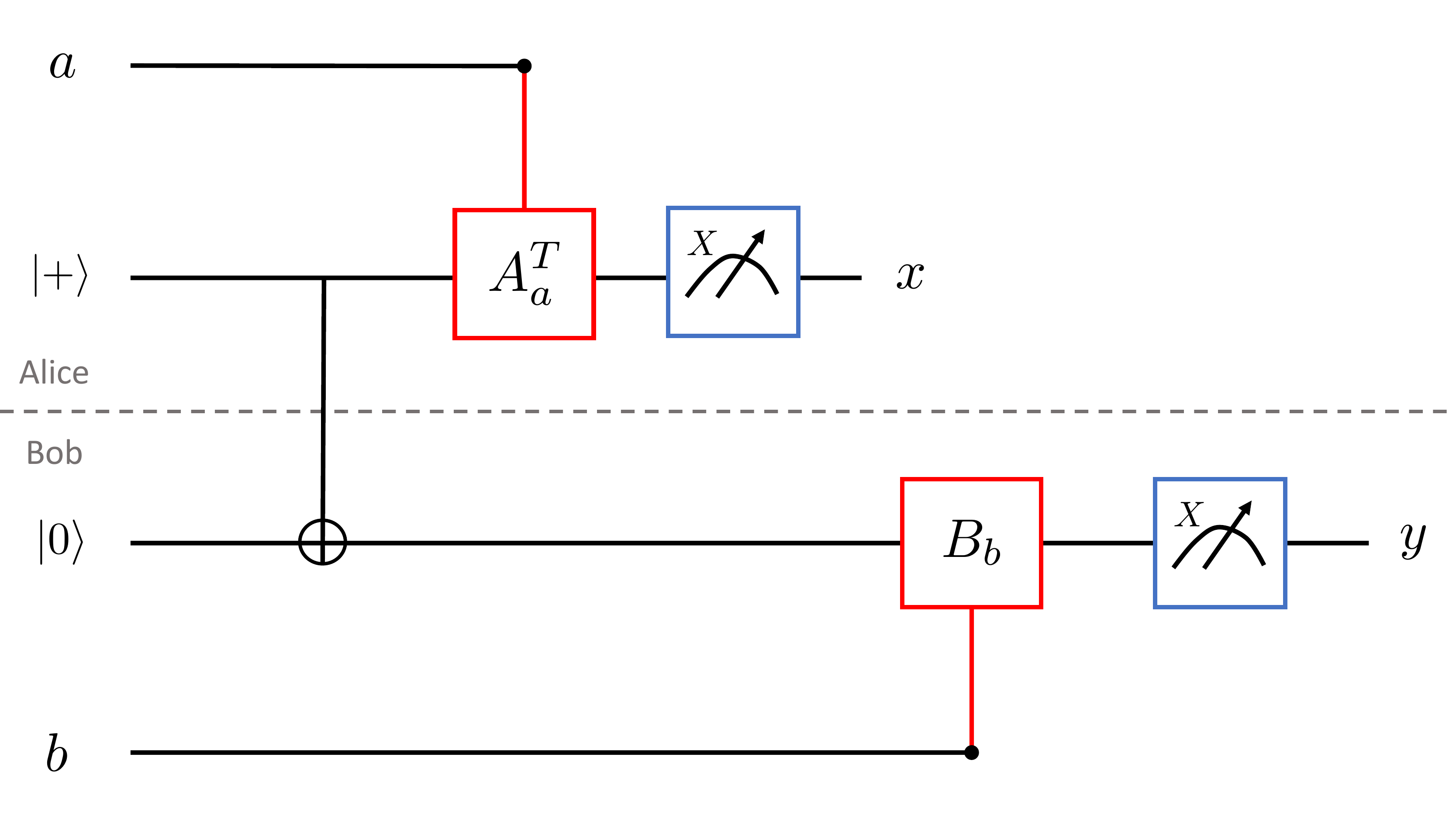}}
\caption{\footnotesize{\textbf{Mapping of the CHSH* game to the CHSH game.} Fig.~\ref{MappingSingle} shows the single qubit scheme, with the initial qubit in state $\left|+\right\rangle,$ controlled gates $A_a,B_b,$ measurement on the $X$ basis and output $c.$ Figure \ref{MappingCHSH} shows the corresponding CHSH game, where Alice and Bob share a Bell pair, and apply gates $A^T_a,B_b$ to their systems to obtain measurement results $x$ and $y$ respectively.}}
\label{CHSH}
\end{figure*}

\textit{Further settings.}
The  proof of Lemma~\ref{MainLemma} relies on the fact that the transformations are unitary, and that the system in the CHSH* game has dimension 2. We will now study the game in other settings, and see that its value is strongly setting-dependent.

First, we relax the restriction that transformations must be unitary by considering the \textit{irreversible setting}. We now allow irreversible transformations, such as the $\mathsf{ERASE}$ map, which maps any qubit state to the state $\ket{0}.$ This may be achieved via a $Z$ measurement and conditional $X$ correction. Introducing irreversible transformations has a dramatic effect on the value of the CHSH* game.

\begin{Prop}\label{IrrevTheorem}
The value of the CHSH* game with a $d=2$ classical or quantum system in the irreversible setting is 1.
\end{Prop}

\begin{proof}
Proof is via explicit example. 	Let the initial state be $\ket{0}$ and let $A_0=\mathbb{I},~ A_1=X,~B_0=\mathsf{ERASE}, B_1=\mathbb{I}$. The final measurement is in the $Z$ basis. Considering the different input cases, we see that the output $c$ will always be $0$ unless both $a$ and $b$ are 1. Thus this strategy always wins the game. Every element of the strategy presented in this proof can be achieved in a classical system, hence we can conclude that this maximum value of 1 can be achieved even with no quantum dynamics at all.
\end{proof}

This increase in the value of the game depends crucially on the \textit{irreversibility} of the $\mathsf{ERASE}$ map. As we see directly, if we restrict logic operations to be reversible, we find that the value of the game is reduced.

\begin{Prop}\label{rev}
The value of the CHSH* game with a $d=2$ classical system in the reversible setting is 0.75.
\end{Prop}

\begin{proof}
To show that the value is at least 0.75, it suffices to describe a protocol which attains this success probability. This is given by the trivial protocol where the input bit is set to 0 and gates $A_a$ and $B_b$ are the identity, and thus the output is always 0. To see why this cannot be exceeded, we observe that all reversible one-bit functions are linear functions. The closest linear function to $a\cdot b$ is the constant function $f(a,b)=0$.
\end{proof}

To summarise the results so far, we have studied the  CHSH* game with a variety of restrictions on the system, which we called settings. We have found values of the game of 0.75, $\cos^2(\frac{\pi}{8})$ and 1, depending on the setting. These precisely match the Bell bound, Tsirelson bound and PR-box value of the CHSH game.

We now show that the CHSH* game is sensitive to further restrictions. Motivated by the crucial role that Clifford/non-Clifford group operations play in quantum computation \cite{Gottesman99}, we now study systems where the operations are restricted to Clifford unitaries. We will address these systems as being in a \textit{Clifford setting}, i.e. the initial system is a pure stabilizer state, all transformations are unitary Cliffords and the measurement is a Pauli observable. Recall that stabilizer states are eigenstates of Pauli operators and that the Clifford gates are gates that map stabilizer states to stabilizer states.

 \begin{Prop}\label{Clifford}
The value of the CHSH* game with a $d=2$ quantum system in the Clifford setting is 0.75.
\end{Prop}

\begin{proof}
Since both transformations are Cliffords, the state $B_bA_a\ket{+}$ is an eigenstate of a Pauli operator, and therefore when measured with the Pauli $X$ operator, will always yield one of the possible outcomes with probability $0$, $0.5$ or $1$. From definition, the average probability of success of the game $\omega(\textrm{CHSH*})$ can therefore obtain one of eight possible values in $\{0,\frac{1}{8},\dots, \frac{7}{8},1\}$.
As we have shown in Proposition \ref{MainTheorem}, when we consider all unitary operations, $\omega(\textrm{CHSH*})=\cos^2{(\frac{\pi}{8})}\approx 0.85$ , which is strictly smaller than $\frac{7}{8}$. In this (restricted to only Cliffords) setting, we can conclude that the maximum attainable value of $\omega(\textrm{CHSH*})$ is $\frac{3}{4}=0.75$.
\end{proof}

We see that restricting the CHSH* game to the Clifford setting gives a success probability equal to the reversible classical setting. This, again, resembles the CHSH game, where if states,  operations and measurements are similarly limited, the Bell inequality value of 0.75 cannot be surpassed. We now show that when diagonal non-Clifford gates are available, one can always do better than this bound. 

\begin{Prop}\label{Cliffordplus}
For a quantum system with $d=2$, in the Clifford setting but with the addition of any pair of non-Clifford gates $R_z(\varepsilon)$ and $R_z(\varepsilon)^\dagger$, with $\varepsilon\in(0,\frac{\pi}{2})$, the value of the CHSH* game is greater than $0.75$.
\end{Prop}

\begin{proof}
The proof is via explicit construction. We adopt a strategy similar to the optimal quantum strategy in the unitary setting, where replacing $T$ with $R_z(\varepsilon)$ and $T^{\dagger}$ with $R^{\dagger}_z(\varepsilon),$ achieves a probability of success $P_{\textrm{suc}}$ greater than $0.75$:
\[\begin{split}P_{\textrm{suc}} & =\frac{1}{4}\biggl[\left(\frac{1}{2}+\frac{\cos(\varepsilon)}{2}\right)+\left(\frac{1}{2}+\frac{\cos(-\varepsilon)}{2}\right) \\ & +\left(\frac{1}{2}+\frac{\cos(\frac{\pi}{2}-\varepsilon)}{2}\right)+\left(1-\frac{1}{2}-\frac{\cos(\frac{\pi}{2}+\varepsilon)}{2}\right)\biggr].\end{split}\]
This probability is always greater than $0.75$ when  $\varepsilon\in(0,\frac{\pi}{2}),$ and attains a maximum of $\cos^2(\frac{\pi}{8})$ when $\varepsilon=\frac{\pi}{4}$ as expected.
\end{proof}

{\noindent{}Figure \ref{QubitStrategies} provides a geometrical comparison of optimal strategies in the three reversible settings we have considered. }

\begin{figure}[t!] 
\includegraphics[width=.49\textwidth]{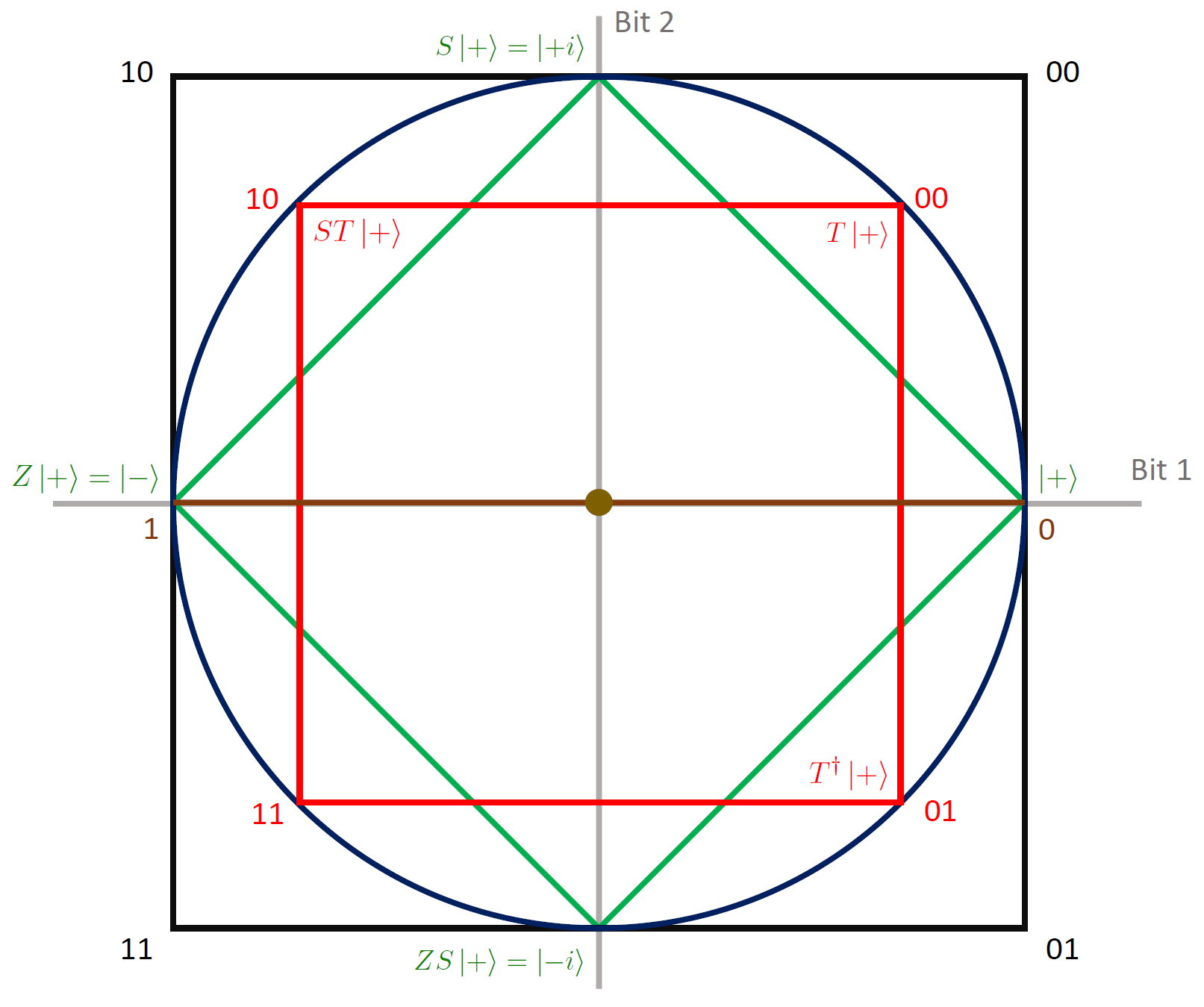}
\caption{\footnotesize{
\textbf{Geometrical analysis of the protocol} 
The figure shows the state space of two bits (vertices of the big black square), one qubit ($XY$ plane of the Bloch sphere) both in the optimal winning strategy (the vertices of the red square) and restricted to Clifford computation (the vertices of the tilted green square), and one bit, (\emph{e.g.} the edges of the brown line). Notice that the measurement at the end of the protocol corresponds to the collapse of a state to the $x$ axis.
}}

\label{QubitStrategies}
\end{figure}

Having seen that the value of the CHSH* game allows us to distinguish between various settings with systems of dimension $2$, we will now consider systems of higher dimension, beginning with dimension $3$.

\begin{Prop}\label{Qutrit}
For d-dimensional quantum or classical systems, in the reversible setting with $d\ge 3,$ there always exists a perfect strategy (i.e. the value of the game is 1).
\end{Prop}

\begin{proof}
We provide a qutrit strategy, and note that this can always be embedded into systems of dimension greater than $3$.
Without loss of generality we suppose that the system is prepared in the state $\left |0 \right\rangle$, and the strategy consists of the gates $A_0=\mathbb{I},A_1=X,B_0=\mathbb{I},B_1=X$. The generalised Pauli $X$ acts as $X\left |i \right\rangle=\left |i +1\right\rangle,$ where $i\in\{0,1,2\}$ and the sum is $\bmod{3}$. The measurement is given by the PVM $\{\left|0\right\rangle\left\langle 0\right| + \left|1\right\rangle\left\langle 1\right| , \left|2\right\rangle\left\langle 2\right| \}$. If we associate the outcome $0$ to the first element of the measurement and the outcome $1$ to the second, we obtain $a\cdot b\bmod{2}$ with probability $1.$ Notice that this strategy can equally be applied in the case of a classical trit, using the obvious analogous state and reversible gates. 
\end{proof} 

This shows that, if the operations on the system are restricted to reversible gates, the CHSH* game acts similar to a \emph{dimensional witness}, as it can witness when the dimension of the system is at least $3$. This is closely related to the seminal work of \cite{Gallego2010}, especially the first case study concerning witnesses for classical bits and qubits.

\textit{Connection to Landauer's principle.} We have seen that erasure is a powerful tool that allows to win the CHSH* game with certainty. Reversible classical and quantum settings lead to distinct lower values for the game, which can be used to identify the nature and dimension of the system; this gives us a new perspective on the non-classical nature of quantum information storage and measurement. 

It was first argued by Landauer \cite{Landauer1961} that irreversible operations are not fundamental. Landauer's principle states that every irreversible classical operation on logical bits must be accompanied by a rise in the entropy of the non-information-bearing degrees of the system or its environment. This holds because in order to build an irreversible gate out of fundamentally reversible operations, we need to discard or erase information. Following Landauer's approach, we associate the erasure of a single bit with an increase in entropy of $kT \log_2 2$, where $k$ is the Boltzmann constant and $T$ the temperature of the system and environment.

The optimal strategy in the irreversible setting, which wins the game with certainty, requires erasure in only one of the four input combinations of $a$ and $b$ (i.e. when $a=1$ and $b=0$). The average increase in entropy is therefore $\frac{1}{4}kT \log_2 2$. Now, let us consider the same strategy, but now implementing a \emph{partial} erasure for the same input combination; when $a=1$ and $b=0$, the system is erased with probability $\sqrt{2}-1\approx 0.41$ and therefore the increase in entropy is on average $\frac{1}{4}0.41kT \log_2 2$. This leads to a success probability equal to the one obtained using a quantum system (i.e. $\cos^2(\frac{\pi}{8})$) for the CHSH$^*$ game, but this time using a classical system and partial erasure.

We can interpret the success probability as how much the chosen setting allows us to learn about the irreversible function $a\cdot b.$ Since we can always transform an irreversible setting to a reversible one by increasing the amount of memory, what the quantum bound exhibits is the qubit's ability to simulate two classical bits (one of which is going to be erased). This is made even more explicit in Figure \ref{QubitStrategies}, which compares the state spaces of a pair of bits, a single qubit and a single bit. In particular, in the optimal quantum strategy the single-qubit state space encodes the four possible input combinations as four quantum states. The measurement then extracts one bit of information.
Since the four states are not all {pairwise} orthogonal, the system is not storing two independent bits prior to the measurement and can therefore perform better than the reversible classical and Clifford settings.

\textit{Generalisation to higher dimensions.}
We have introduced the CHSH* game as a modification of the CHSH game from two players to one player. It is natural to  consider a similar one-player modification of the $\bmod{q}$ CHSH$_q$ game. We call such a game the CHSH$_q^*$ game. An interesting question is whether Lemma~\ref{MainTheorem} can be extended to a correspondence between strategies for the single qudit and CHSH$_q$ games.
The current proof of the lemma does not directly generalise to systems of higher dimension since it utilises some special properties of 2x2 unitary matrices.
 
Nevertheless, we conjecture that the correspondence between the Tsirelson bound for the CHSH game and the quantum value for the CHSH$_q^*$ game in the unitary setting holds for arbitrary dimensions. We here provide a support towards the validity of the conjecture, by focusing on the case of $q=3$. Specifically, we show that a strategy in the CHSH$_3^*$ game mapped from a slight modification of an optimal quantum strategy in the CHSH$_3$ as provided by Ji et al. in \cite{Ji2008}, obtains exactly the value of Tsirelson's bound for the CHSH$_3$ game, which is known to be approximately 0.71 \cite{Buhrman2005, Ji2008, Liang2009, Bavarian2015}. We also show that the Bell bound of $\nicefrac{2}{3}$ for the CHSH$_3$ game holds equally for the CHSH$_3^*$ game.

Since we are in $\bmod 3$ arithmetics, the CHSH$_3^*$ game is won if the player's final measurement outputs $c=a\cdot b\Mod{3}$, for inputs $a,b,c \in\{0,1,2\}$. For a classical trit with reversible gates, the maximum probability of success (coinciding with the known Bell bound \cite{Buhrman2005, Ji2008, Liang2009, Bavarian2015}) is $\nicefrac{2}{3}$. This can be found by listing all the possibilities for the different input values. One way to obtain it is to start with the trit in the state $0$ and apply the gates $A_0=A_1=B_0=B_2=\mathbb{I}, A_2=B_1=X.$ 

Suppose now that we have a qutrit system prepared in state \[T_3\ket{+}=T_3\frac{\ket{0}+\ket{1}+\ket{2}}{\sqrt{3}},\] 
where the gate $T_3=\text{diag}(1,w^{-\nicefrac{1}{3}},w^{-\nicefrac{2}{3}})$ is the dimension-3 equivalent of the non-Clifford gate $T$, and $w=\exp(\frac{2\pi i}{3}).$ 

Let us choose the following control gates: 
\[A_0=B_0=\mathbb{I},A_1=B_2=V, A_2=B_1=W,\] 
where $V=\text{diag}(1,w,w)$ and $W=\text{diag}(1,1,w).$ Measuring the system in the $X$ basis gives a success probability $P_{\textrm{suc}}\approx 0.71$. This strategy is inspired by the one used to obtain the Tsirelson bound for the CHSH$_3$ game in \cite{Ji2008}, thereby providing support for the conjecture that there exists a mapping from the single system protocol to CHSH in higher dimensions.

\textit{Conclusion.}
In this work we introduced the CHSH* game, a single player game inspired by the CHSH game. We showed that the optimal success probability for CHSH*, called the  value of the game, depends on many properties of the system available to the players. Defining these properties via settings, we showed that the value of the game depends on the irreversibility, or otherwise, of the transformations available to the players, the quantum or classical nature of the system and the system dimension.

Furthermore, we saw that the values obtained are equal to the Bell and Tsirelson bounds in the CHSH game (and the perfect strategies embodied by PR boxes). In particular, for the unitary quantum setting, Lemma \ref{MainTheorem} shows that any unitary strategy in CHSH* can be mapped to a quantum strategy in the CHSH game. This correspondence gives a new perspective on Tsirelson's bound, which arises due to the absence of irreversible transformations and the limited ability of quantum strategies with unitary gates and projective measurements to simulate erasure. In the more restricted Clifford setting, the value of the game does not exceed the reversible classical setting, hinting a connection with the Gottesman-Knill theorem that states that circuits consisting of gates from the Clifford group can be efficiently simulated by a probabilistic classical computer. This reflects the crucial role of non-Clifford computation to obtain better-than-classical performance in quantum computation. 

We show that, under the assumption of reversible transformations, the CHSH* game resembles a dimensional witness, since any initial state residing in a state space of dimension $d>2$ can in principle win the game with certainty. However, the restriction to reversible operations is not a limitation. In accordance with Landauer's principle, implementing irreversible transformations at the microscopic level requires ancillary bits which must then be erased. The presence of exactly these hidden ancilliary bits is detected by our protocol.

We noted a similarity between the optimal unitary strategy for the CHSH* game and quantum Random Access Codes (RAC). The latter have also been proposed as dimensional witnesses \cite{Wehner2006}. It is therefore important to emphasise the differences between RAC and the CHSH* game. The CHSH* game is able to detect the hidden information needed to implement irreversible gates. However, irreversible gates provide no advantage for the implementation of Random Access Codes. This means that a dimensional witness based on the RAC protocol will be blind to this kind of hidden information. Following Landauer's approach, we assert that the ability to detect irreversible dynamics should be an important desideratum for quantum dimensional witnesses. 

We conjecture our results to hold also for the generalisation of the protocol to $\bmod{q}$ arithmetics. We support this by examining the $q=3$ case in the single system scenario, for which we show the validity of the Bell bound and we further provide a strategy to achieve Tsirelson's bound. The validity of this conjecture may open the way to easier approaches for deriving Tsirelson's bounds in $\bmod{q}$ arithmetics, by using our single-system protocol as a tool for proving tightness.

This work also demonstrates how Tsirelson's bound arises from the restricted physics of a unitarily evolving single qubit system. In light of Landauer's principle, we considered the entropic costs of the erasure associated with the CHSH* game and how the lack of such an operation in unitary quantum mechanics is a barrier to winning the game deterministically. Via the correspondence with Tsirelson's bound proven in Lemma~\ref{MainLemma}, we demonstrate a link between the reversibility in fundamental operations embodied by Landauer's principle, and the non-unity value of Tsirelson's bound. In this way, we complement previous studies that have established connections of thermodynamics with uncertainty relations \cite{hanggi2013} as well as with the different interpretations of quantum theory \cite{Cabello2016}.

Finally, a recent paper \cite{MansfieldKashefi2018} has introduced a new notion of transformation contextuality, where the contexts are sequences of transformations in a $l2$-TBQC protocol. This work is relevant to the CHSH* game, since  \cite[Theorem~1]{MansfieldKashefi2018} applies to the CHSH* game too. Other forms of contextuality have been studied from the single-particle perspective \cite{Spekkens2009}, but they do not apply here. Our work shows that assumptions of reversibility in transformations can have a dramatic effect on the capabilities of the system, motivating further study of the relationship between non-classicality and irreversible dynamics.

\textit{Acknowledgements.} We would like to thank Carlo Sparaciari for useful discussions on Landauer's principle. This work was supported by the Brazilian National Council for Research and Development (CNPq) [206604/2014-9], the EPSRC Centre for Doctoral Training in Delivering Quantum Technologies [EP/L015242/1], A.P.~and S.M.~acknowledge support from European Union's Horizon 2020 Research and Innovation program under Marie Sklodowska-Curie Grant Agreements No.~705194 and No.~750523, respectively.
\\

\noindent{}Luciana Henaut and Lorenzo Catani contributed equally to this work.

\end{document}